\documentclass[lettersize,journal]{IEEEtran}
\usepackage{amsmath,amsfonts}
\usepackage{algorithm}
\usepackage{algpseudocode}
\usepackage{array}
\usepackage[caption=false,font=normalsize,labelfont=sf,textfont=sf]{subfig}
\usepackage{textcomp}
\usepackage{stfloats}
\usepackage{url}
\usepackage{verbatim}
\usepackage{booktabs}
\usepackage{graphicx}
\usepackage{cite}
\usepackage[table]{xcolor}
\usepackage{multirow} 

\usepackage[acronym]{glossaries}
\makeglossaries
\newacronym{ce}{CE}{Capture Effect}
\newacronym[plural=APs]{ap}{AP}{Access Point}
\newacronym[plural=BSSs]{bss}{BSS}{Basic Service Set}
\newacronym{sr}{SR}{Spatial Reuse}
\newacronym{obsspd}{OBSS/PD}{Overlapping Basic Service Set/Packet Detect}
\newacronym[plural=STAs]{sta}{STA}{Station}
\newacronym{cosr}{Co-SR}{Coordinated Spatial Reuse}
\newacronym{cotdma}{Co-TDMA}{Coordinated Time Division Multiple Access}
\newacronym{mapc}{MAPC}{Multi-Access Point Coordination}
\newacronym{txop}{TXOP}{Transmission Opportunity}
\newacronym{mcs}{MCS}{Modulation and Coding Scheme}
\newacronym{sinr}{SINR}{Signal-to-Interference-plus-Noise Ratio}
\newacronym{pf}{PF}{Proportional Fairness}

\usepackage{soul}

\usepackage{amsthm}
\newtheorem{theorem}{Theorem}[section]

\definecolor{staticclr}{RGB}{222,222,222}
\definecolor{donorclr}{RGB}{255,214,178}
\definecolor{contclr}{RGB}{190,220,255}
\definecolor{discclr}{RGB}{200,235,200}

\begin{document}

\bstctlcite{IEEEexample:BSTcontrol}

\title{On the Allocation of Transmit Power for Coordinated Spatial Reuse in IEEE 802.11bn Multi-Access Point Coordination}

\author{Francesc~Wilhelmi,~\IEEEmembership{Member,~IEEE} and Boris Bellalta,~\IEEEmembership{~Senior Member,~IEEE}
\thanks{The authors are with Universitat Pompeu Fabra, Barcelona, Spain.}}

\maketitle

\begin{abstract}
IEEE 802.11bn (11bn) introduces Coordinated Spatial Reuse (Co-SR), a Multi-Access Point Coordination (MAPC) scheme where two APs coordinate to control the transmit power for a simultaneous transmission. This letter proposes a Proportional Fairness (PF)-driven framework for allocating Co-SR transmit power on a per-Transmission Opportunity (TXOP) basis. We prove that any Pareto-optimal power pair keeps at least one AP at its maximum power, reducing the joint two-dimensional search to two cheap one-dimensional line searches that fit within TXOP timing, and show that this collapses to an exact closed-form expression in a discrete-rate system. We validate these policies and a low-complexity selfish baseline against a brute-force oracle, using packet-level simulations in \texttt{Kom8ndor}, an 11bn simulator. Results show that jointly evaluating both APs' links is key to unlocking Co-SR's gains, that the coordinated AP's fairness-optimal power depends critically on whether the rate is modeled continuously or through discrete Modulation and Coding Scheme (MCS) steps, and that current 11bn signaling supports the selfish policy but not the fairness-optimal ones, which would need additional per-TXOP channel reporting.
\end{abstract}


\section{Introduction}

\IEEEPARstart{T}{he} upcoming IEEE 802.11bn (11bn) standard (Wi-Fi~8) is almost ready and includes new features such as \gls{mapc}~\cite{wilhelmi2026tutorial}. Among other schemes, \gls{cosr} is introduced to enable simultaneous transmissions between two \glspl{bss} on the same channel, while controlling the transmit power used by each one to ensure that each receiver's \gls{sinr} remains acceptable. This is the natural successor to the \gls{obsspd}-based \gls{sr} of IEEE 802.11ax (11ax)~\cite{wilhelmi2021spatial}, but now includes explicit inter-\gls{ap} signaling to better control interference with respect to decentralized energy-detection thresholding.

While defining the procedures by which \glspl{ap} exchange information and coordinate simultaneous transmissions, 11bn leaves open the implementation of transmit power adjustment. Prior work has explored this in different ways. Some implementations assume either symmetric power allocation or restrict only the power of the coordinated \gls{ap}, hence disregarding interference geometry~\cite{wilhelmi2023throughput}. Regarding dynamic approaches, \cite{wojnar2025coordinated} proposed a hierarchical multi-armed bandit approach for \gls{cosr} parameter selection and scheduling. In a related direction, joint \gls{ap} selection and power control via multi-agent deep reinforcement learning was proposed in~\cite{jung2025coordinated}. Finally, centralized approaches for \gls{ap} power control and \gls{txop}-sharing resource allocation were studied in \cite{haxhibeqiri2024coordinated} and \cite{nunez2022txop}, respectively. These approaches optimize learned or heuristic policies, evaluated offline rather than derived in closed form. This leads to the following questions: \emph{What is the optimal joint power configuration for the two \glspl{ap}?} \emph{Can the optimal values be computed cheaply to match \gls{txop} duration times?} And finally, \emph{Is the existing 11bn \gls{mapc} signaling enough to achieve that?}


This letter answers these questions from a \gls{pf} perspective: (i) we show that the optimal transmit power pair always has at least one \gls{ap} at its power ceiling (a structural result for the sum-rate objective~\cite{gjendemsjo2008binary}, adapted here to \gls{pf}~\cite{kelly1998rate}), (ii) we reduce the remaining one-dimensional problem to a cheap line search, (iii) we show that a discrete-rate system collapses this search to an exact closed form, (iv) we validate the continuous variant numerically against a brute-force oracle (the discrete variant's optimality follows from monotonicity, requiring no such validation), and (v) we implement and evaluate both variants in \texttt{Kom8ndor}~\cite{wilhelmi2026kom8ndor}.

\section{System Model}
\label{sec:sysmodel}

We consider the scenario in Fig.~\ref{fig:scenario}, which comprises a coordinator and a coordinated \gls{ap} that transmit simultaneously to their associated \gls{sta} in a shared \gls{txop}. \gls{mapc} negotiation is performed per \gls{txop}, as shown in Fig.~\ref{fig:signaling}~\cite{wilhelmi2026tutorial}.

\begin{figure}
    \centering
    \subfloat[]{\includegraphics[width=.70\columnwidth]{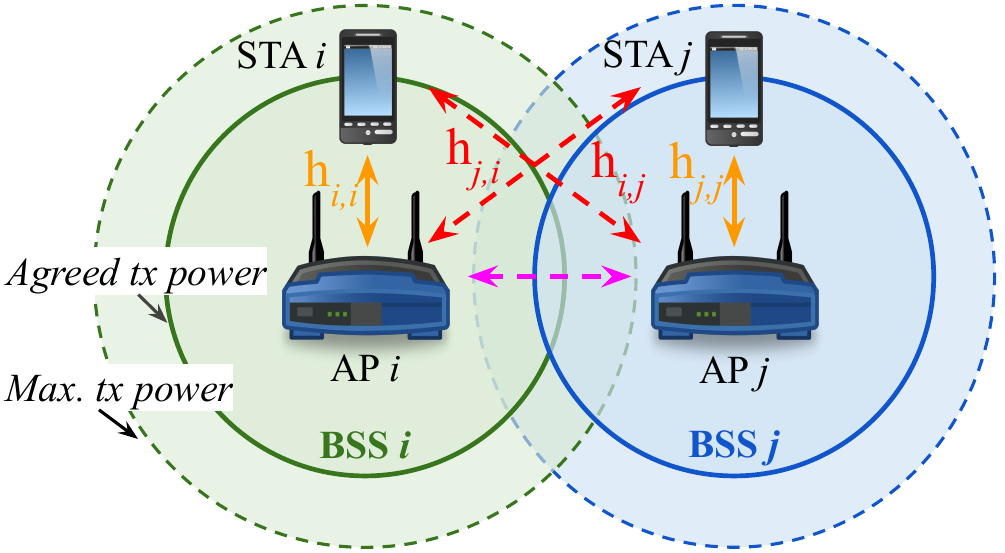}\label{fig:scenario}}    
    \hfil
    \subfloat[]{\includegraphics[width=.70\columnwidth]{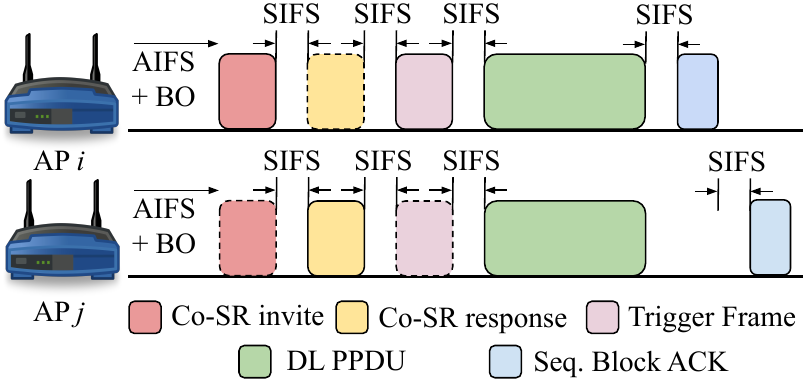}\label{fig:signaling}}
    \caption{Co-SR scenario. (a) Deployment. (b) MAPC Co-SR signaling.}
    \label{fig:example}
\end{figure}

\emph{Received power and channel gain.} Using a standard log-distance pathloss model, the power received at \gls{sta}$_i$ from $\mathrm{AP}_k\in\{i,j\}$ is
\begin{equation}
\label{eq:prx}
P_{k,i}^{\text{rx}} = P_k\, G_{\text{tx}} G_{\text{rx}} \, L(d_{k,i})^{-1},
\end{equation}
where $P_k$ is $\mathrm{AP}_k$'s transmit power, $G_{\text{tx}}$ and $G_{\text{rx}}$ are the (linear) transmit and receive antenna gains, respectively, and $L(d)$ is the path loss at distance $d$. The channel gain is 
\begin{equation}
\label{eq:chg}
    h_{k,i} = P_{k,i}^{\text{rx}}/P_k = G_{\text{tx}} G_{\text{rx}}L(d_{k,i})^{-1}.
\end{equation}

\emph{SINR and rate.} During the simultaneous transmission, from the perspective of a given \gls{bss} $i$, we treat the transmissions of the other \gls{bss} $j$ as noise, so the received \gls{sinr} at \gls{sta} $i$ is
\begin{equation}
\label{eq:sinr}
\textrm{SINR}_i(P_i, P_j) = \frac{h_{i,i} P_i}{N_0 + h_{j,i} P_j},
\end{equation}
where $N_0$ is the floor noise and $P_i,P_j\in(0,P_{\max}]$ are the transmit powers of \glspl{ap} $i,j$ (to be optimized), capped at $P_{\max}$. All four channel gains $\{h_{i,i},h_{i,j},h_{j,i},h_{j,j}\}$ are assumed to be known thanks to \gls{mapc} signaling (an idealization that is currently partially supported). For the analysis in Section~\ref{sec:pf}, we adopt the Shannon rate $R_i(P_i,P_j)=\log_2(1+\textrm{SINR}_i(P_i,P_j))$ [b/s/Hz] as a tractable proxy for the achievable rate. 

\emph{MCS and decoding.} In practice, \glspl{ap}' transmission rate depends on a discrete \gls{mcs} $m_i=\mu(P_{i,i}^{\text{rx}})$ that depends on the power received (11bn supports up to 20 different \gls{mcs} types). We assume that the \gls{mcs} is selected independently of the interference perceived by \gls{bss} $j$, following the minimum receiver sensitivity thresholds of the IEEE 802.11 standard~\cite{wilhelmi2026kom8ndor}. Regarding decoding, a transmission in \gls{bss} $i$ is successful if $\textrm{SINR}_i(P_i, P_j) \ge \gamma_{\text{CE}}$, where $\gamma_{\text{CE}}$ is the (linear) \gls{ce} target used by the receiver's decoder, assumed here to be independent of \gls{mcs} selection.


\section{Power Allocation Strategies for \gls{cosr}}
\label{sec:pf}

We consider four power allocation strategies for \gls{cosr} (see Algorithm~\ref{alg:cosr}): \colorbox{staticclr}{static} ($P_i=P_j=P_{\max}$), \colorbox{donorclr}{selfish} (included as a low complexity baseline, Section~\ref{ssec:donor}), \colorbox{contclr}{continuous} (Section~\ref{ssec:linesearch}), and \colorbox{discclr}{discrete} (Section~\ref{ssec:discrete}). The three latter strategies build on the formulation of the \gls{pf}~\cite{kelly1998rate}:
\begin{equation}
\begin{aligned} \max_{P_i, P_j} \quad & \text{PF}(P_i, P_j) = \ln R_i + \ln R_j \\ \text{subject to} \quad & 0 < \{P_i, P_j\} \le P_{\max} . \end{aligned}
\end{equation}

We index the two boundary edges by $k\in\{1,2\}$ (edge 1 fixes $P_i{=}P_{\max}$, free $P_j$; edge 2 fixes $P_j{=}P_{\max}$, free $P_i$), and write $\text{PF}(P)=\ln R_{\text{fixed}}+\ln R_{\text{free}}(P)$ for the \gls{pf} objective as a function of the free \gls{ap}'s power on that edge.

\begin{algorithm}[h]
\caption{Per-TXOP \gls{cosr} power allocation, unifying \colorbox{staticclr}{static}, \colorbox{donorclr}{selfish}, \colorbox{contclr}{continuous}, and \colorbox{discclr}{discrete} strategies.}
\label{alg:cosr}
\begin{algorithmic}[1]
\Require $h_{ii}, h_{ij}, h_{ji}, h_{jj}$, $P_{\max}$, $\gamma_{\text{CE}}$
\If{\colorbox{staticclr}{static}} \Return $(P_{\max}, P_{\max})$ \EndIf
\State $\mathcal{E} \gets \{1\}$ if \colorbox{donorclr}{selfish}, else $\{1,2\}$
\For{$k \in \mathcal{E}$} \State Compute $P_k^{\text{floor}},P_k^{\text{ceil}}$ using Eq.~\eqref{eq:floorceil}
  \If{$P_k^{\text{floor}} \le P_k^{\text{ceil}}$}
    \State $P^\star_k, \text{PF}_k \gets$ \Call{OptimizeEdge}{$P_k^{\text{floor}}$, $P_k^{\text{ceil}}$}
    \State $(P_i,P_j)^{(k)} \gets (P_{\max},P^\star_k)$ if $k{=}1$, else $(P^\star_k,P_{\max})$
  \EndIf
\EndFor
\If{no feasible $k \in \mathcal{E}$} \Return \textsc{Co-TDMA} \EndIf
\State \Return $(P_i,P_j)^{(k^\star)}$, $k^\star = \arg\max_k \text{PF}_k$
\Statex
\Function{OptimizeEdge}{$P_k^{\text{floor}}$, $P_k^{\text{ceil}}$}
  \If{\colorbox{contclr}{continuous}}
    \For{$n=1,\dots,N$} \Comment{\colorbox{contclr}{ternary search}}
      \State \colorbox{contclr}{$\Delta \gets (P_k^{\text{ceil}}{-}P_k^{\text{floor}})/3$}
      \State \colorbox{contclr}{$t_1,t_2 \gets P_k^{\text{floor}}{+}\Delta,\ P_k^{\text{ceil}}{-}\Delta$}
      \State \colorbox{contclr}{$P_k^{\text{floor}}\gets t_1$ if $\text{PF}(t_1){<}\text{PF}(t_2)$, else $P_k^{\text{ceil}}\gets t_2$}
    \EndFor
    \State \colorbox{contclr}{$P^\star \gets (P_k^{\text{floor}}+P_k^{\text{ceil}})/2$; \Return $P^\star, \text{PF}(P^\star)$}
  \ElsIf {\colorbox{discclr}{discrete} or \colorbox{donorclr}{selfish}}
    \State $P^\star \gets \min(P_{\max}, P_k^{\text{ceil}})$; \Return $P^\star, \text{PF}(P^\star)$
  \EndIf
\EndFunction
\end{algorithmic}
\end{algorithm}

\subsection{Preliminaries}
\label{ssec:prelim}

Before delving into the considered transmit power allocation methods, we first look at relevant aspects for deriving them. First, finding the optimal allocation of $(P_i,P_j)$ requires a 2D search $(0,P_{\max}]\times(0,P_{\max}]$, which can be computationally expensive between \glspl{txop} (around 5~ms). For that reason, we apply the Boundary Theorem to reduce the search to 1D~\cite{gjendemsjo2008binary}.

\begin{theorem}
\label{thm:boundary}
At any Pareto-optimal $(P_i, P_j)$, $\max(P_i, P_j) = P_{\max}$.
\end{theorem}

\begin{proof}
Fix the ratio $\rho = P_i/P_j$ and scale both powers by a common factor $\lambda>0$ ($P_i(\lambda)=\rho \lambda P_j$, $P_j(\lambda)=\lambda P_j$). Then
\begin{equation}
\mathrm{SINR}_i(\lambda) = \frac{h_{ii}\rho P_j}{N_0/\lambda + h_{ji} P_j},
\label{eq:sinr_proof}
\end{equation}
which is strictly increasing in $\lambda$ for any fixed $\rho > 0$. Therefore, both $R_i$ and $R_j$ are non-decreasing in $\lambda$ at fixed $\rho$, so scaling $\lambda$ up to the feasibility boundary (when $\max(P_i, P_j)$ reaches $P_{\max}$) Pareto-dominates any interior point. So, the true optimum is guaranteed to be on either $P_i=P_{\max}$ or $P_j=P_{\max}$, as Eq.~\ref{eq:sinr_proof} is symmetric for $\mathrm{SINR}_j(\lambda)$. \end{proof}

Adopting Theorem~\ref{thm:boundary} reduces the 2D search to two 1D searches with no loss of optimality, and holds unconditionally for any rate function not decreasing in \gls{sinr} (not only Shannon rate)~\cite{gjendemsjo2008binary}. Another consideration of the proposed methods lies in the decodability capabilities of the receivers in a \gls{cosr} transmission. In particular, optimizing \gls{pf} alone might lead to failures, as the \gls{sinr} of a given receiver might fall below the \gls{ce} $\gamma_\text{CE}$. As a result, for $P_i=P_{\max}$ (interchangeably for $P_j=P_{\max}$), $P_j$ must satisfy

\begin{equation}
\label{eq:floorceil}
\underbrace{\frac{\gamma_{\text{CE}}(N_0+h_{ij}P_{\max})}{h_{jj}}}_{P_j^{\text{floor}}\;(j\text{ decodes})}
\le P_j \le
\underbrace{\frac{h_{ii}P_{\max}/\gamma_{\text{CE}} - N_0}{h_{ji}}}_{P_j^{\text{ceil}}\;(i\text{ still decodes})}.
\end{equation}

With this, we restrict the search for the transmit power of one \gls{ap} to this interval. Moreover, if no $P_j$ satisfies the conditions in Eq.~\eqref{eq:floorceil}, meaning the \gls{ce} condition does not hold for at least one \gls{bss} (leading to packet losses), we fall back to \gls{cotdma} with equal time split among the two \glspl{ap}. Note that Theorem~\ref{thm:boundary} holds independently of the \gls{ce}-based decodability model, since its proof relies only on rate monotonicity in own \gls{sinr}. Therefore, fixing $\gamma_{\text{CE}}$ affects only the specific bounds in Eq.~\eqref{eq:floorceil} and the resulting closed form of Method~3 (Sec.~\ref{ssec:discrete}).

\subsection{Method 1: Selfish}
\label{ssec:donor}

A cheaper alternative to the strategies that follow assumes that the coordinating \gls{ap} ($i$, the winner of the \gls{txop}) shares its \gls{txop} by using the maximum transmit power itself and limiting the transmit power of the coordinated \gls{ap} ($j$) based on the maximum interference \gls{sta} ($\mathrm{STA}_i$) can tolerate. Setting $P_i=P_{\max}$ in Eq.~\eqref{eq:sinr} and requiring $\textrm{SINR}_i\ge\gamma_{\text{CE}}$ gives the maximum interference budget
\begin{equation}
\label{eq:imax}
I_i^{\max} = \frac{h_{ii}P_{\max}}{\gamma_{\text{CE}}} - N_0,
\end{equation}
which leads to
\begin{equation}
\label{eq:donor}
P_j = \min\!\Big(P_{\max},\ \frac{I_i^{\max}}{h_{ji}}\Big) = \min(P_{\max}, P_j^{\text{ceil}}) .
\end{equation}

This approach is computationally cheap but sacrifices optimality, since it does not check whether letting the coordinated \gls{ap} keep $P_{\max}$ would be better. In addition, it is reasonable that the \gls{ap} that wins a \gls{txop} does not sacrifice its performance by reducing its power to benefit possibly unknown \glspl{ap}~\cite{wilhelmi2023throughput}. 

\subsection{Method 2: Continuous Search}
\label{ssec:linesearch}

Theorem~\ref{thm:boundary} reduces the two-dimensional problem to two one-dimensional edge problems: $(P_{\max}, P_j)$ and $(P_i, P_{\max})$, where $P_j$ and $P_i$ are independently optimized, respectively, over the \gls{ce}-feasible range shown in Sec.~\ref{ssec:prelim}. On either edge, the rate of the free \gls{ap} (the one whose transmit power is to be optimized) strictly increases with its power, while the fixed \gls{ap}'s rate strictly decreases due to the added interference. Consistent with this, we verified numerically (Sec.~\ref{ssec:analytical}) that the resulting \gls{pf} objective ($\hat{\text{PF}}= \ln R_{\text{fixed}}+\ln R_{\text{free}}$) is unimodal in the free power. To derive the \gls{cosr} transmit power through the continuous line search mechanism, we apply a ternary search, costing only $2N$ evaluations per edge. In every iteration $n=1,\dots,N$, the current interval of candidate power values for the free \gls{ap} (initialized to $[P_k^{\text{floor}}, P_k^{\text{ceil}}]$) is evaluated at the two points splitting it into thirds. The portion that cannot contain the maximum is discarded. The policy evaluates both edges this way and returns the one that yields higher \gls{pf}.

\subsection{Method 3: Discrete Rate}
\label{ssec:discrete}

The continuous line search optimizes the analytical expression $R_i=\log_2(1+\mathrm{SINR}_i)$, but the real system's rate is based on the \gls{mcs} ($m_i=\mu(P_{i,i}^{\text{rx}})$), which is computed from the received power. Once both links fulfill their \gls{ce} condition, link $i$'s achieved rate depends only on $P_i$. In consequence, since the rate function is non-decreasing in own received power and independent of the peer's power once the \gls{ce} constraint from Eq.~\eqref{eq:floorceil} is satisfied, the edge optimum is

\begin{equation}
\label{eq:closedform}
P_{k}^\star = \min\!\big(P_{\max},\, P_{k}^{\text{ceil}}\big).
\end{equation}

This means that the \gls{ap}'s power to be optimized should be pushed to the top of its \gls{ce}-feasible range, so no oracle validation is needed for the discrete-rate method. 


\section{Evaluation}
\label{sec:eval}

We validate that Theorem~\ref{thm:boundary}'s search is lossless and evaluate the different mechanisms using \texttt{Kom8ndor}~\cite{wilhelmi2026kom8ndor}. Table~\ref{tab:params} shows the main simulation parameters.

\begin{table}[t!]
\centering
\caption{Main simulation parameters.}
\label{tab:params}
\resizebox{\columnwidth}{!}{%
\begin{tabular}{ll}
\toprule
\textbf{Parameter} & \textbf{Value} \\
\midrule
Frequency band ($f_c$) & 5\,GHz \\
Channel width ($B$) & 20\,MHz \\
Path loss model ($PL$) & IEEE TGax Scenario 1~\cite{merlin2015tgax} \\
Tx power ceiling ($P_{\max}$) & 20\,dBm \\
Noise floor ($N_0$) & $-95$\,dBm \\
Tx/Rx antenna gains ($G_{tx}/G_{rx}$) & 0\,dB (omnidirectional) \\
CCA thresh. ($CCA$) & $-82$\,dBm \\
CE ($\gamma_{\text{CE}}$) & 10\,dB \\
MCS values ($m$) & BPSK 1/2 to 4096-QAM 5/6 \\
Traffic model ($\zeta$) & Full buffer (downlink) \\
Max. A-MPDU aggregation ($N_{agg}$) & 64 \\
Max. TXOP duration ($T_\text{TXOP}$) & 5.484\,ms \\
Ternary search iterations ($N$) & 40 \\
Simulation time ($T_\text{sim}$) & 100\,s \\
\bottomrule
\end{tabular}
}
\end{table}

\subsection{Analytical validation}
\label{ssec:analytical}

We first verify that the two-edge search (Theorem~\ref{thm:boundary}) is lossless for the continuous mechanism from Algorithm~\ref{alg:cosr}. We consider a deployment with two \glspl{ap} ($\mathrm{AP}_i$, coordinating, and $\mathrm{AP}_j$, coordinated), each serving one \gls{sta} located 3\,m and 1\,m, respectively, from its \gls{ap}. We evaluate the continuous search method and compare it against a brute-force search (oracle) and the static policy over inter-\gls{ap} distances $D\in\{2,3,4,5,6\}$\,m (at $D=1$\,m both edges are empty and the policy falls back to \gls{cotdma}). The results (power, rate, and resulting \gls{pf}) are reported in Table~\ref{tab:analytical}. As shown, the continuous search mechanism matches the oracle (within $10^{-3}$) at roughly $1/1000$ of the oracle's computational cost.\footnote{For the oracle, we considered a $400\times400=160{,}000$-point grid, while the continuous search runs $2\times2N=160$ for $N{=}40$.} We also observe that the continuous search improves the static policy by 0.03--0.16 throughout the different distances in terms of \gls{pf}. Remarkably, this is achieved by uplifting the rate of \gls{ap}$_i$, which is disadvantaged relative to \gls{ap}$_j$ in this scenario.

\begin{table}[h!]
\centering
\caption{Analytical validation of the continuous search method.}
\label{tab:analytical}
\resizebox{.85\columnwidth}{!}{%
\begin{tabular}{ccccc}
\toprule
$D\,[m]$ & Policy & $(P_i,P_j)$ [dBm] & $(R_i,R_j)$ [b/s/Hz] & $\text{PF}$ \\
\midrule
\multirow{3}{*}{2} & \cellcolor{staticclr}static     & \cellcolor{staticclr}(20.0, 20.0) & \cellcolor{staticclr}(2.20, 5.47) & \cellcolor{staticclr}2.489 \\
 & \cellcolor{contclr}continuous & \cellcolor{contclr}(20.0, 14.6) & \cellcolor{contclr}(3.75, 3.76) & \cellcolor{contclr}2.646 \\
 & \textit{oracle} & \textit{(20.0, 14.6)} & \textit{(3.75, 3.76)} & \textit{2.646} \\
\midrule
\multirow{3}{*}{3} & \cellcolor{staticclr}static     & \cellcolor{staticclr}(20.0, 20.0) & \cellcolor{staticclr}(3.70, 8.63) & \cellcolor{staticclr}3.464 \\
 & \cellcolor{contclr}continuous & \cellcolor{contclr}(20.0, 12.5) & \cellcolor{contclr}(6.10, 6.15) & \cellcolor{contclr}3.624 \\
 & \textit{oracle} & \textit{(20.0, 12.5)} & \textit{(6.09, 6.15)} & \textit{3.623} \\
\midrule
\multirow{3}{*}{4} & \cellcolor{staticclr}static     & \cellcolor{staticclr}(20.0, 20.0) & \cellcolor{staticclr}(5.45, 11.38) & \cellcolor{staticclr}4.129 \\
 & \cellcolor{contclr}continuous & \cellcolor{contclr}(20.0, 11.4) & \cellcolor{contclr}(8.23, 8.54) & \cellcolor{contclr}4.252 \\
 & \textit{oracle} & \textit{(20.0, 11.5)} & \textit{(8.22, 8.55)} & \textit{4.252} \\
\midrule
\multirow{3}{*}{5} & \cellcolor{staticclr}static     & \cellcolor{staticclr}(20.0, 20.0) & \cellcolor{staticclr}(7.53, 13.79) & \cellcolor{staticclr}4.643 \\
 & \cellcolor{contclr}continuous & \cellcolor{contclr}(20.0, 12.3) & \cellcolor{contclr}(9.93, 11.24) & \cellcolor{contclr}4.715 \\
 & \textit{oracle} & \textit{(20.0, 12.3)} & \textit{(9.93, 11.24)} & \textit{4.715} \\
\midrule
\multirow{3}{*}{6} & \cellcolor{staticclr}static     & \cellcolor{staticclr}(20.0, 20.0) & \cellcolor{staticclr}(9.41, 16.20) & \cellcolor{staticclr}5.027 \\
 & \cellcolor{contclr}continuous & \cellcolor{contclr}(20.0, 14.5) & \cellcolor{contclr}(10.96, 14.38) & \cellcolor{contclr}5.060 \\
 & \textit{oracle} & \textit{(20.0, 14.5)} & \textit{(10.96, 14.38)} & \textit{5.060} \\
\bottomrule
\end{tabular}
}
\end{table}

\subsection{Simulation Results}
\label{ss:simulation_results}

We simulate the same deployment from Section~\ref{ssec:analytical} using \texttt{Kom8ndor} and evaluate each considered transmit power selection policy (\colorbox{staticclr}{static}, \colorbox{donorclr}{selfish}, \colorbox{contclr}{continuous}, and \colorbox{discclr}{discrete}), to which we add a non-\gls{mapc} baseline (\emph{no-mapc}). The results for $D\in \{1,2,3,4,5,6\}$\,m are provided in Table~\ref{tab:simulation_results_1}.

\begin{table}[t!]
\centering
\caption{Simulation results.}
\label{tab:simulation_results_1}
\resizebox{\columnwidth}{!}{%
\begin{tabular}{cccccc}
\toprule
$D\,[m]$ & Policy & $(P_i,P_j)$ [dBm] & $\hat{\text{PF}}$ & $(\text{Thr}_i,\text{Thr}_j)$ [Mbps] & $\text{PF}$ \\
\midrule
\multirow{5}{*}{1} & \textit{no-mapc}   & \textit{(20.0, 20.0)} & -- & \textit{(44.49, 66.07)} & \textit{7.986} \\
 & \cellcolor{staticclr}static     & \cellcolor{staticclr}(20.0, 20.0) & \cellcolor{staticclr}-- & \cellcolor{staticclr}(0.00, 0.00) & \cellcolor{staticclr}$-\infty$ \\
 & \cellcolor{donorclr}selfish     & \cellcolor{donorclr}(Co-TDMA) & \cellcolor{donorclr}-- & \cellcolor{donorclr}(38.09, 59.70) & \cellcolor{donorclr}7.729 \\
 & \cellcolor{contclr}continuous   & \cellcolor{contclr}(Co-TDMA) & \cellcolor{contclr}-- & \cellcolor{contclr}(38.09, 59.70) & \cellcolor{contclr}7.729 \\
 & \cellcolor{discclr}discrete     & \cellcolor{discclr}(Co-TDMA) & \cellcolor{discclr}-- & \cellcolor{discclr}(38.09, 59.70) & \cellcolor{discclr}7.729 \\
\midrule
\multirow{5}{*}{2} & \textit{no-mapc}   & \textit{(20.0, 20.0)} & -- & \textit{(41.78, 70.41)} & \textit{7.987} \\
 & \cellcolor{staticclr}static     & \cellcolor{staticclr}(20.0, 20.0) & \cellcolor{staticclr}2.489 & \cellcolor{staticclr}(0.00, 128.26) & \cellcolor{staticclr}$-\infty$ \\
 & \cellcolor{donorclr}selfish     & \cellcolor{donorclr}(20.0, 15.06/20.00) & \cellcolor{donorclr}2.644 & \cellcolor{donorclr}(60.61, 125.08) & \cellcolor{donorclr}8.933 \\
 & \cellcolor{contclr}continuous   & \cellcolor{contclr}(20.0, 14.61) & \cellcolor{contclr}2.645 & \cellcolor{contclr}(85.04, 127.56) & \cellcolor{contclr}9.292 \\
 & \cellcolor{discclr}discrete     & \cellcolor{discclr}(20.0, 15.06) & \cellcolor{discclr}2.644 & \cellcolor{discclr}(85.04, 127.56) & \cellcolor{discclr}9.292 \\
\midrule
\multirow{5}{*}{3} & \textit{no-mapc}   & \textit{(20.0, 20.0)} & -- & \textit{(41.82, 70.34)} & \textit{7.987} \\
 & \cellcolor{staticclr}static     & \cellcolor{staticclr}(20.0, 20.0) & \cellcolor{staticclr}3.464 & \cellcolor{staticclr}(85.04, 127.56) & \cellcolor{staticclr}9.292 \\
 & \cellcolor{donorclr}selfish     & \cellcolor{donorclr}(20.0, 20.0) & \cellcolor{donorclr}3.464 & \cellcolor{donorclr}(85.04, 127.56) & \cellcolor{donorclr}9.292 \\
 & \cellcolor{contclr}continuous   & \cellcolor{contclr}(20.0, 12.48) & \cellcolor{contclr}3.624 & \cellcolor{contclr}(85.04, 127.56) & \cellcolor{contclr}9.292 \\
 & \cellcolor{discclr}discrete     & \cellcolor{discclr}(20.0, 20.0) & \cellcolor{discclr}3.464 & \cellcolor{discclr}(85.04, 127.56) & \cellcolor{discclr}9.292 \\
\midrule
\multirow{5}{*}{4} & \textit{no-mapc}   & \textit{(20.0, 20.0)} & -- & \textit{(41.82, 70.34)} & \textit{7.987} \\
 & \cellcolor{staticclr}static     & \cellcolor{staticclr}(20.0, 20.0) & \cellcolor{staticclr}4.129 & \cellcolor{staticclr}(85.04, 127.56) & \cellcolor{staticclr}9.292 \\
 & \cellcolor{donorclr}selfish     & \cellcolor{donorclr}(20.0, 20.0) & \cellcolor{donorclr}4.129 & \cellcolor{donorclr}(85.04, 127.56) & \cellcolor{donorclr}9.292 \\
 & \cellcolor{contclr}continuous   & \cellcolor{contclr}(20.0, 11.43) & \cellcolor{contclr}4.252 & \cellcolor{contclr}(85.04, 127.56) & \cellcolor{contclr}9.292 \\
 & \cellcolor{discclr}discrete     & \cellcolor{discclr}(20.0, 20.0) & \cellcolor{discclr}4.129 & \cellcolor{discclr}(85.04, 127.56) & \cellcolor{discclr}9.292 \\
\midrule
\multirow{5}{*}{5} & \textit{no-mapc}   & \textit{(20.0, 20.0)} & -- & \textit{(47.12, 69.87)} & \textit{8.099} \\
 & \cellcolor{staticclr}static     & \cellcolor{staticclr}(20.0, 20.0) & \cellcolor{staticclr}4.643 & \cellcolor{staticclr}(85.04, 127.56) & \cellcolor{staticclr}9.292 \\
 & \cellcolor{donorclr}selfish     & \cellcolor{donorclr}(20.0, 20.0) & \cellcolor{donorclr}4.643 & \cellcolor{donorclr}(85.04, 127.56) & \cellcolor{donorclr}9.292 \\
 & \cellcolor{contclr}continuous   & \cellcolor{contclr}(20.0, 12.31) & \cellcolor{contclr}4.715 & \cellcolor{contclr}(85.04, 127.56) & \cellcolor{contclr}9.292 \\
 & \cellcolor{discclr}discrete     & \cellcolor{discclr}(20.0, 20.0) & \cellcolor{discclr}4.643 & \cellcolor{discclr}(85.04, 127.56) & \cellcolor{discclr}9.292 \\
\midrule
\multirow{5}{*}{6} & \textit{no-mapc}   & \textit{(20.0, 20.0)} & -- & \textit{(46.99, 70.22)} & \textit{8.102} \\
 & \cellcolor{staticclr}static     & \cellcolor{staticclr}(20.0, 20.0) & \cellcolor{staticclr}5.027 & \cellcolor{staticclr}(85.11, 127.66) & \cellcolor{staticclr}9.293 \\
 & \cellcolor{donorclr}selfish     & \cellcolor{donorclr}(20.0, 20.0) & \cellcolor{donorclr}5.027 & \cellcolor{donorclr}(85.11, 127.66) & \cellcolor{donorclr}9.293 \\
 & \cellcolor{contclr}continuous   & \cellcolor{contclr}(20.0, 14.53) & \cellcolor{contclr}5.060 & \cellcolor{contclr}(85.11, 127.66) & \cellcolor{contclr}9.293 \\
 & \cellcolor{discclr}discrete     & \cellcolor{discclr}(20.0, 20.0) & \cellcolor{discclr}5.027 & \cellcolor{discclr}(85.11, 127.66) & \cellcolor{discclr}9.293 \\
\bottomrule
\end{tabular}%
}
\end{table}

For $D=1$\,m, \emph{no-mapc} outperforms every \gls{cosr} variant, as coordination signaling from the adaptive policies (which fall back to \gls{cotdma}) reduces efficiency. The \emph{static} policy, which has no feasibility check and forces simultaneous transmissions at 20~dBm, leads to zero performance due to the high interference regime it creates. At $D=2$\,m, all the adaptive policies start to adapt the power to enable concurrent transmissions. At this distance, the \emph{selfish} policy falls short due to its single-edge restriction and only performs optimally during the \glspl{txop} won by \gls{ap}$_i$, when \gls{ap}$_j$'s power is limited to 15.06\,dBm (optimal allocation requires the weakest link to receive maximum power). The rest of the time, when the stronger link (\gls{ap}$_j$-\gls{sta}$_j$) wins and the maximum transmit power is allocated to it, simultaneous transmission is unfeasible, and the method falls back to \gls{cotdma} (note that the $\hat{\text{PF}}$ value reported in that case refers to the case when \gls{ap}$_i$ wins the \gls{txop}). From $D=3$\,m onward, when space can be properly reused due to the geometry of the scenario, all adaptive policies except \emph{continuous} converge to the highest rates, i.e., $(85.04, 127.56)$\,Mbps, by selecting the maximum power at both \glspl{ap}. \emph{Continuous} matches the same \gls{pf} as the other policies while using considerably less power at \gls{ap}$_j$.


\section{Conclusions}
\label{sec:conclusions}

We studied the optimization of the transmit power in 11bn \gls{cosr} transmissions. We advocated for \gls{pf} as a maximization goal, as we believe that unlicensed wireless networks should be altruistic by nature, but other optimization objectives (e.g., sum throughput) can be considered. Analytical and simulation results revealed: 1) that the two-edge search is lossless, allowing for reducing costly grid searches to 1D searches and unlocking cheap (within \gls{txop} time) computation, 2) that, to find the optimal allocation (where one \gls{ap} uses $P_{\max}$ but the other \gls{ap}'s power rarely uses the maximum) both links must be evaluated, and 3) that the existing 11bn \gls{mapc} signaling, which includes the allowed coordinated \gls{ap} power (\texttt{TX Power Limit} field), is sufficient for the selfish mechanism, but includes no per-\gls{txop} report of the coordinated \gls{ap}'s own link quality, which the double-edge approaches require and which \gls{cosr} signaling does not currently define. As future work, we plan to extend the foundations laid in this paper by adding more realistic decoding capabilities and \gls{mcs} adaptation in a broader set of evaluation scenarios.


\section*{Acknowledgements}

Supported by TRUE Wi-Fi PID2024-155470NB-I00 (MICIU/AEI/10,13039/501100011033/FEDER,UE), ICREA Academia 2024 (00077 AGAUR), and MdM CEX2021-001195-M (MICIU/AEI/10.13039/501100011033).

\bibliographystyle{IEEEtran}
\bibliography{references}

\end{document}